\documentclass[pra,aps,showpacs,twocolumn,twoside,superscriptaddress]{revtex4}

\usepackage{amsmath,amsfonts,amssymb,color,epsfig,graphics,graphicx,latexsym,revsymb,theorem,verbatim}

\newtheorem{definition}{Definition}
\newtheorem{proposition}[definition]{Proposition}
\newtheorem{lemma}[definition]{Lemma}

\newtheorem{theorem}[definition]{Theorem}
\newtheorem{corollary}[definition]{Corollary}
\newtheorem{conjecture}[definition]{Conjecture}

\newtheorem{remark}[definition]{Remark}

\def\squareforqed{\hbox{\rlap{$\sqcap$}$\sqcup$}}
\def\qed{\ifmmode\squareforqed\else{\unskip\nobreak\hfil
\penalty50\hskip1em\null\nobreak\hfil\squareforqed
\parfillskip=0pt\finalhyphendemerits=0\endgraf}\fi}
\def\endenv{\ifmmode\;\else{\unskip\nobreak\hfil
\penalty50\hskip1em\null\nobreak\hfil\;
\parfillskip=0pt\finalhyphendemerits=0\endgraf}\fi}
\newenvironment{proof}{\noindent \textbf{{Proof.~} }}{\qed}

\def\bcj{\begin{conjecture}}
\def\ecj{\end{conjecture}}
\def\bcr{\begin{corollary}}
\def\ecr{\end{corollary}}
\def\bd{\begin{definition}}
\def\ed{\end{definition}}
\def\bea{\begin{eqnarray}}
\def\eea{\end{eqnarray}}
\def\bem{\begin{enumerate}}
\def\eem{\end{enumerate}}
\def\bim{\begin{itemize}}
\def\eim{\end{itemize}}
\def\bl{\begin{lemma}}
\def\el{\end{lemma}}
\def\bpf{\begin{proof}}
\def\epf{\end{proof}}
\def\bpp{\begin{proposition}}
\def\epp{\end{proposition}}
\def\br{\begin{remark}}
\def\er{\end{remark}}
\def\bt{\begin{theorem}}
\def\et{\end{theorem}}

\def\h{\eta}
\def\t{\theta}

\def\r{\rho}

\def\ph{\varphi}

\def\ps{\psi}

\def\Ps{\Psi}

\newcommand{\nc}{\newcommand}

\nc{\cA}{{\cal A}} \nc{\cB}{{\cal B}} \nc{\cC}{{\cal C}}
\nc{\cD}{{\cal D}} \nc{\cE}{{\cal E}} \nc{\cF}{{\cal F}}
\nc{\cG}{{\cal G}} \nc{\cH}{{\cal H}} \nc{\cI}{{\cal I}}
\nc{\cJ}{{\cal J}} \nc{\cK}{{\cal K}} \nc{\cL}{{\cal L}}
\nc{\cM}{{\cal M}} \nc{\cN}{{\cal N}} \nc{\cO}{{\cal O}}
\nc{\cP}{{\cal P}} \nc{\cR}{{\cal R}} \nc{\cS}{{\cal S}}
\nc{\cT}{{\cal T}} \nc{\cX}{{\cal X}} \nc{\cZ}{{\cal Z}}

\def\ghz{\mathop{\rm GHZ}}

\def\rk{\mathop{\rm rk}}

\def\tr{\mathop{\rm Tr}}
\def\w{\mathop{\rm W}}

\def\bigox{\bigotimes}
\def\dg{\dagger}

\def\ox{\otimes}
\def\ra{\rightarrow}
\def\Ra{\Rightarrow}

\newcommand{\ket}[1]{|#1\rangle}
\newcommand{\proj}[1]{| #1\rangle\!\langle #1 |}
\newcommand{\ketbra}[2]{|#1\rangle\!\langle#2|}

\newcommand{\jmp}{J. Math. Phys.}

\begin{document}
\title{Multi-copy and stochastic transformation of multipartite pure states}

\author{Lin Chen}
\email{cqtcl@nus.edu.sg (Corresponding~Author)}
\author{Masahito Hayashi$^{2,}$}
\email{hayashi@math.is.tohoku.ac.jp}
\address{Centre for Quantum Technologies, National University of
Singapore, 3 Science Drive 2, 117542, Singapore\\
$^2$ Graduate School of Information Sciences, Tohoku
University, Aoba-ku, Sendai, 980-8579, Japan}

\begin{abstract}
Characterizing the transformation and classification of multipartite
entangled states is a basic problem in quantum information. We study
the problem under two most common environments, local operations and
classical communications (LOCC), stochastic LOCC and two more
general environments, multi-copy LOCC (MCLOCC) and multi-copy SLOCC
(MCSLOCC). We show that two transformable multipartite states under
LOCC or SLOCC are also transformable under MCLOCC and MCSLOCC.
What's more, these two environments are equivalent in the sense that
two transformable states under MCLOCC are also transformable under
MCSLOCC, and vice versa. Based on these environments we classify the
multipartite pure states into a few inequivalent sets and orbits,
between which we build the partial order to decide their
transformation. In particular, we investigate the structure of
SLOCC-equivalent states in terms of tensor rank, which is known as
the generalized Schmidt rank. Given the tensor rank, we show that
GHZ states can be used to generate all states with a smaller or
equivalent tensor rank under SLOCC, and all reduced separable states
with a cardinality smaller or equivalent than the tensor rank under
LOCC. Using these concepts, we extended the concept of "maximally
entangled state" in the multi-partite system.

\end{abstract}

\date{\today}


\maketitle

\section{\label{sec:introduction} introduction}

%

One of the main tasks in quantum information theory is to find out
how many different ways there exist, in which several spatially
distributed objects could be entangled under certain prior
environments with restricted physical resource. For example, we
often use local operations assisted with classical communication
(LOCC) with which one can obtain quantum resources definitely. This
is a key condition to transform and prepare different multipartite
states from each other, which are the basic ingredient in
quantum-information tasks, such as GHZ states in quantum
teleportation \cite{bbc93} and graph states in quantum computation
\cite{br01}. It is known that bipartite pure entangled states are
interconvertible in the asymptotic LOCC transformations
\cite{bbp96}. In other words, all bipartite pure entangled states
can be used to perform the same quantum-information task in the
asymptotic regime.

On the other hand, the problem becomes complex for multipartite
states. For example, it is known that two pure states are
transformable from each other with certainty, namely
\emph{equivalent} if and only if they are related by local unitary
operation (LU), i.e., $\ket{\ps} = \bigox_i U_i \ket{\ph}$ where
$U_i$ are unitary operators \cite{bpr00}. However in this way we can
remove only very few parameters that characterize the system. As a
result, we have to compare quantum states with an exponentially
increasing number of parameters and it is thus unlikely to get a
clear judgement of their convertibility under LOCC. To reduce the
difficulty, Bennett et al \cite{bpr00} introduced another
environment, namely the conversion of states through stochastic LOCC
(SLOCC) with a nonzero success probability. Two pure states of a
multipartite system are equivalent under SLOCC if and only if they
are related by an invertible local operator, i.e., $\ket{\ps} =
\bigox_i A_i \ket{\ph}$ where $A_i$ are invertible operators
\cite{dvc00}. For instance, concerning the three-qubit fully
entangled states which contains no factorized systems under SLOCC,
there exist only two different families namely the
Greenberger-Horne-Zeilinger (GHZ) state \cite{ghz89} and the W state
\cite{dvc00}. In principle, the total classification of $2\times M
\times N$ states under SLOCC is realizable in terms of
hyper-determinant \cite{miyake03}, range criterion \cite{ccm06} and
matrix pencil~\cite{cmy10}. Furthermore, with SLOCC researchers have
also addressed the behavior of multiqubit states, which are the
fundamental resource in quantum computation and communications. In
addition, the SLOCC environment is also a key technique to a few
other important quantum information tasks, such as entanglement
distillation~\cite{bds96}.

In this paper we generalize the above two environments, which focus
on one-copy transformation of quantum states. We propose the
multi-copy LOCC (MCLOCC) and multi-copy SLOCC (MCSLOCC) in the sense
that one can use many copies of states to produce a target state
under LOCC and SLOCC, respectively \cite{brs02}. In fact, when a
state $|\varphi\rangle$ can be transformed to $|\psi\rangle$ by
SLOCC, we need to prepare many copies of $|\varphi\rangle$ for
generating the state $|\psi\rangle$. So, from the operational
viewpoint, it is natural to consider the convertibility from plural
copies of the given state $|\varphi\rangle$ For example, the
transformation between bipartite pure states is subjective to
Nielsen's majorization theorem~\cite{nielsen99}, and states
violating the theorem are usually not inter-convertible. However as
already mentioned~\cite{bbp96}, we can produce any target state with
enough copy of input states. The fact indicates that there is an
essential difference between one-copy and multi-copy transformation.
This is important for various quantum-information tasks which may
require different kinds of states as the input states. By studying
MCLOCC and MCSLOCC, we can more effectively make use of the quantum
resources. In particular, we will show that the two environments are
equivalent in the sense that two transformable states under MCLOCC
are also transformable under MCSLOCC, and vice versa. This is
different from the relation between LOCC and SLOCC environment,
where two equivalent states under SLOCC are often inequivalent under
LOCC.

Based on LOCC, SLOCC and MCLOCC environments, we build a hierarchy
to classify the multipartite pure states. We show that two
transformable states under LOCC are also transformable under SLOCC,
and further transformable under MCLOCC. In this sense, the LOCC
environment is superior to other two environments. To characterize
the transformation under SLOCC, we introduce the concept of tensor
rank, which is also known as the Schmidt measure of
entanglement~\cite{eb01}. As tensor rank is an invariant under
invertible SLOCC, in terms of it and the local ranks of reduced
density operators we can further characterize the structure of
multipartite states under SLOCC.

We also propose the concept of \emph{maximally} entangled states in
under a given environment, in the sense that it can be used to
generate all states in a set or orbit. First, we will show that any
fully entangled $n$-partite pure state is maximal concerning MCLOCC
among all $n$-partite pure states. Furthermore, for a given tensor
rank, any state SLOCC-equivalent to GHZ states is a maximally
entangled state concerning SLOCC in the set of states with the
tensor rank. Third, we introduce the concept cardinality to describe
the smallest number of product states forming the convex sum of a
separable state, whose purification we name as the reduced separable
pure states. Then we show that the GHZ state with rank $d$ is
maximal concerning LOCC among reduced separable pure states with the
largest cardinality $d$.



The rest of the paper is organized as follows. In
Sec.~\ref{sec:definition} we define the basic environment for state
transformation including multi-copy LOCC (MCLOCC), SLOCC and LOCC.
We also establish their relations such as partial order and
equivalence. In Sec.~\ref{sec:classificationMCLOCC},
\ref{sec:classificationSLOCC}, and \ref{sec:classificationLOCC} we
classify the multipartite pure states by three elementary steps in
terms of MCLOCC, SLOCC and LOCC, respectively. We conclude in
Sec.~\ref{sec:conclusion}.

\section{\label{sec:definition} definitions and general classification}

We start by defining a few useful criterion for one-way
transformation between multipartite pure states.

 \bd
We denote $|\psi_1\rangle \succ_{LOCC} |\psi_2\rangle$ when there
exists LOCC that transforms $|\psi_1\rangle$ to $|\psi_2\rangle$. We
denote $|\psi_1\rangle \succ_{SLOCC} |\psi_2\rangle$ when there
exists matrices $A_1,\ldots, A_n$ such that $A_1\otimes \ldots
\otimes A_n |\psi_1\rangle = |\psi_2\rangle$.
 \ed

 \bd
We denote $|\psi_1\rangle \succ_{MCLOCC} |\psi_2\rangle$ (multi-copy
LOCC) when there exists integer $n$ such that
$|\psi_1\rangle^{\otimes n} \succ_{LOCC} |\psi_2\rangle$. We denote
$|\psi_1\rangle \succ_{MCSLOCC} |\psi_2\rangle$ (multi-copy SLOCC)
when there exists integer $n$ such that $|\psi_1\rangle^{\otimes n}
\succ_{SLOCC} |\psi_2\rangle$.
 \ed

Trivially, when $|\psi_1\rangle \succ_{LOCC} |\psi_2\rangle$, we
have $|\psi_1\rangle \succ_{SLOCC} |\psi_2\rangle$. Although the
converse does not generally hold (e.g., $\cos \t \ket{00} + \sin \t
\ket{11} \succ_{SLOCC} (\ket{00} + \ket{11})/\sqrt2$ but the
transformation $\cos \t \ket{00} + \sin \t \ket{11} \succ_{LOCC}
(\ket{00} + \ket{11})/\sqrt2$ is prohibited by the majorization
criterion \cite{nielsen99}), we have a further result that when
$|\psi_1\rangle \succ_{SLOCC} |\psi_2\rangle$, $|\psi_1\rangle
\succ_{MCLOCC} |\psi_2\rangle$. This is because when there exists
POVM elements $A_1, \cdots, A_n$ which can be used to do the
transformation $A_1\otimes \ldots \otimes A_n |\psi_1\rangle =
|\psi_2\rangle$ with nonzero probability $p>0$. Then we can repeat
the measurement many times until we succeed with the probability
$1-(1-p)^n,~ n \ra \infty$. In this sense, we can deterministically
transform the state $|\psi_1\rangle$ into $|\psi_2\rangle$ with
sufficiently large $n$.

However, the converse does not hold in general. That is, there are
states $\ket{\ps_1}, \ket{\ps_2}$ such that when $|\psi_1\rangle
\succ_{MCLOCC} |\psi_2\rangle$, but not $|\psi_1\rangle
\succ_{SLOCC} |\psi_2\rangle$. A typical example is the tripartite
GHZ and W state
 \bea
 \ket{\ghz}
 &=&
 \frac{1}{\sqrt2} \bigg( \ket{000} + \ket{111} \bigg),
 \nonumber\\
 \ket{\w}
 &=&
 \frac{1}{\sqrt3} \bigg( \ket{001} + \ket{010} + \ket{100} \bigg),
 \eea
which are known to be inequivalent under reversible SLOCC
\cite{dvc00}. However, it is a necessary and sufficient condition
that two-copy GHZ states are able to generate one W state under LOCC
\cite{ds09}. That is,
$$|\ghz\rangle^{\otimes 2} \succ_{LOCC}
|\w\rangle.$$

Next, we can characterize the relation between the environment of
MCLOCC and MCSLOCC. By using the same argument for deriving
$|\psi_1\rangle \succ_{SLOCC} |\psi_2\rangle ~\Ra~ |\psi_1\rangle
\succ_{MCLOCC} |\psi_2\rangle$, we can get $\ket{\ps_1}
\succ_{MCSLOCC} \ket{\ps_2} ~\Ra~ \ket{\ps_1} \succ_{MCLOCC}
\ket{\ps_2}$. On the other hand, suppose we have $\ket{\ps_1}
\succ_{MCLOCC} \ket{\ps_2}$. Generally we can suppose $\ket{\ps_1}
=\bigox_i \ket{\ph_{i1}},~\ket{\ph_{i1}} \in \cH_i$ are fully
entangled. According to the definition of $MCLOCC$, we can get the
expression $\ket{\ps_2} =\bigox_i \ket{\ph_{i2}},~\ket{\ph_{i2}} \in
\cH_i$. In other words we have
$$\ket{\ps_1} \succ_{MCLOCC} \ket{\ps_2}
~\Ra~ \ket{\ph_{i1}} \succ_{MCLOCC} \ket{\ph_{i2}}$$. Let $\cH_i =
\bigox^{n_i}_{j=1} \cH_{i,j}$, we can use the states
$\ket{\ph_{i1}}$ to get bell states in the bipartite space
$\cH_{i,1} \ox \cH_{i,j}, j=2, \cdots, n_i$ under SLOCC. So we can
use the bell states to generate an arbitrary state $\ket{\ph_{i2}}$
by teleportation. Hence, we get
 \bea
 \ket{\ph_{i1}} \succ_{MCLOCC} \ket{\ph_{i2}}
 &\Ra&
 \ket{\ph_{i1}} \succ_{MCSLOCC} \ket{\ph_{i2}}
 \nonumber\\
 &\Ra&
 \ket{\ps_1} \succ_{MCSLOCC} \ket{\ps_2}.
 \eea
As a short summary of the above arguments, we have
 \bt
 \label{thm:4definitions,oneway}
 For two $n$-partite pure states $\ket{\ps_1}$ and $\ket{\ps_2}$, we have
 \bea
 \ket{\ps_1} &\succ_{LOCC}& \ket{\ps_2}
 \nonumber\\
 &\Downarrow&
 \nonumber\\
 \ket{\ps_1} &\succ_{SLOCC}& \ket{\ps_2}
 \nonumber\\
 &\Downarrow&
 \nonumber\\
 \ket{\ps_1} &\succ_{MCLOCC}& \ket{\ps_2}
 \nonumber\\
 &\Updownarrow&
 \nonumber\\
 \ket{\ps_1} &\succ_{MCSLOCC}& \ket{\ps_2}.
 \eea
 \et

It is noticeable that in the four definitions of
Theorem~\ref{thm:4definitions,oneway},
 only the definition $ \ket{\ps_1} \succ_{MCLOCC} \ket{\ps_2} $ is
 asymptotic while other three are deterministic. In other words, under MCLOCC we
 obtain a state $\r$ such that the fidelity $F(\r, \ket{\ps_2}) \simeq
 1$ when the number of copies $n \ra \infty$. However in all other
 three definitions, we are required to get an exact state
 $\ket{\ps_2}$ from $\ket{\ps_1}^{\ox n}$ with a finite $n$.

Because of Theorem~\ref{thm:4definitions,oneway}, we will use the
definition $\ket{\ps_1} \succ_{MCLOCC} \ket{\ps_2}$ which is equal
to $\ket{\ps_1} \succ_{MCSLOCC} \ket{\ps_2}$ from now on. Based on
the results of one-way transformation, we study the relations under
two-way (invertible) transformation.

 \bd
 \label{def:mcloccequiv}
We denote $|\psi_1\rangle \cong_{MCLOCC} |\psi_2\rangle$, when
$|\psi_1\rangle \succ_{MCLOCC} |\psi_2\rangle$ and $|\psi_2\rangle
\succ_{MCLOCC} |\psi_1\rangle$. So $|\psi_1\rangle $ is called
MCLOCC-equivalent to $|\psi_2\rangle$. We denote the set of states
$|\psi_2\rangle$ by ${\cal O}_{MCLOCC}(|\psi_1\rangle)$, namely the
MCLOCC orbit of state $|\psi_1\rangle$.
 \ed

 \bd
 \label{def:sloccequiv}
Similarly, we denote $|\psi_1\rangle \cong_{SLOCC} |\psi_2\rangle$.
So $|\psi_1\rangle $ is called SLOCC-equivalent to $|\psi_2\rangle$.
We denote as the SLOCC orbit $\cO_{SLOCC}(\ket{\ps_2})$, which is
known to consist of states $\ket{\ps_1} = \bigox_i A_i \ket{\ps_2}$
where $A_i$ are invertible \cite{dvc00}.
 \ed

 \bd
 \label{def:loccequiv}
Similarly, we denote $|\psi_1\rangle \cong_{LOCC} |\psi_2\rangle$.
So $|\psi_1\rangle $ is called LOCC-equivalent to $|\psi_2\rangle$.
We denote as the LOCC orbit $\cO_{LOCC}(\ket{\ps_2})$, which is
known to consist of states $\ket{\ps_1} = \bigox_i U_i \ket{\ps_2}$
where $U_i$ are unitary \cite{bpr00}.
 \ed

Based on these definitions, we can generalize
Theorem~\ref{thm:4definitions,oneway} to the two-way (invertible) case.

 \bt
 \label{thm:4definitions,twoway}
 For two $n$-partite pure states $\ket{\ps_1}$ and $\ket{\ps_2}$, we have
 \bea
 \ket{\ps_1} &\cong_{LOCC}& \ket{\ps_2}
 \nonumber\\
 &\Downarrow&
 \nonumber\\
 \ket{\ps_1} &\cong_{SLOCC}& \ket{\ps_2}
 \nonumber\\
 &\Downarrow&
 \nonumber\\
 \ket{\ps_1} &\cong_{MCLOCC}& \ket{\ps_2}
 \nonumber\\
 &\Updownarrow&
 \nonumber\\
 \ket{\ps_1} &\cong_{MCSLOCC}& \ket{\ps_2}.
 \eea
 \et

Similarly, we will use the definition $\ket{\ps_1} \cong_{MCLOCC}
\ket{\ps_2}$ which is equal to $\ket{\ps_1} \cong_{MCSLOCC}
\ket{\ps_2}$ from now on. We say that two states/orbits are
\textit{X incomparable}, namely $\ket{\ps_1} \ncong_X \ket{\ps_2}$
when the partial order $\succ_X$ does not hold for both directions.

\section{\label{sec:classificationMCLOCC} Classification of multipartite pure states under MCLOCC}

To show the power of definitions for multipartite pure states in
last section, we classify the multipartite pure states in terms of
Theorem~\ref{thm:4definitions,twoway}. Following the three inclusion
relations there, we use three corresponding \emph{elementary} steps
as our strategy. That is,
 \bem
\item
 First, we divide multipartite pure states into a few inequivalent MCLOCC orbits
 and clarify the three partial orders under MCLOCC, SLOCC, and LOCC between these orbits.
\item
 Second, we divide each MCLOCC orbit into a few inequivalent SLOCC
 sets, orbits and clarify their partial order under SLOCC.
\item
 Third, we give a few examples from the SLOCC sets and orbits and study their partial order under LOCC.
 \eem
These steps will be carried out in
Sec.~\ref{sec:classificationMCLOCC}, ~\ref{sec:classificationSLOCC}
and~\ref{sec:classificationLOCC}, respectively.

Let us carry out elementary step 1 which is relatively easy to
finish. For this purpose, we define the concept of the independence
for a given pure state $\ket{\psi}$ on the multipartite system $A_1,
A_2, \ldots ,A_n$. Two systems $A_1$ and $A_2$ are {\it independent}
for a pure multipartite state $\ket{\psi}$ when the density matrix
on the composite system $A$ and $B$ has the form $\rho_A \otimes
\rho_B$. In particular, the two systems $A_1$ and $A_2$ are {\it
completely independent} for $\ket{\psi}$ when the two systems $A_1$
and $A_2$ are independent for any state in the orbit
$\cO_{SLOCC}(\ket{\psi})$.
Two parties $A_1$ and $A_2$ are completely independent for $\ket{\psi}$
if and only if
there are two distinct groups $\{A_{a_1}, \ldots, A_{a_k}\}$ and
$\{A_{b_1}, \ldots, A_{b_l}\}$ among $A_3,\ldots, A_n$
and two pure states
$\ket{\psi_1}$ on $A_1, A_{a_1}, \ldots, A_{a_k}$
and $\ket{\psi_2}$ on $A_2, A_{b_1}, \ldots, A_{b_l}$
such that $\ket{\psi}=\ket{\psi_1}\otimes \ket{\psi_2}$.
Using this relation, we can find that
two systems $A_1$ and $A_2$ are completely independent for a state $\ket{\psi}$
if and only if two systems $A_1$ and $A_2$ are
completely independent for the state $\ket{\psi}^{\otimes n}$.

For a given multipartite state $\ket{\psi}$, we define a graph to
connect two parties that are not independent. We denote such a graph
by $G(\ket{\psi})$. For example, when $\ket{\psi}$ is
$\ket{0000}+\ket{0110}+\ket{1200}+\ket{1310}+\ket{0001}+\ket{0111}+\ket{1201}+\ket{1311}$,
the party $A_1$ is independent of the parties $A_3$ and $A_4$, the
party $A_2$ is independent of the party $A_4$, the party $A_3$ is
independent of the parties $A_1$ and $A_4$, and the party $A_4$ is
independent of the parties $A_1$, $A_2$, and $A_3$. However, the
parties $A_1$ and $A_3$ are not completely independent. We describe
the graph $G(\ket{\psi})$ explicitly in Fig. \ref{fig}.

\begin{figure}
  \includegraphics[width=8cm]{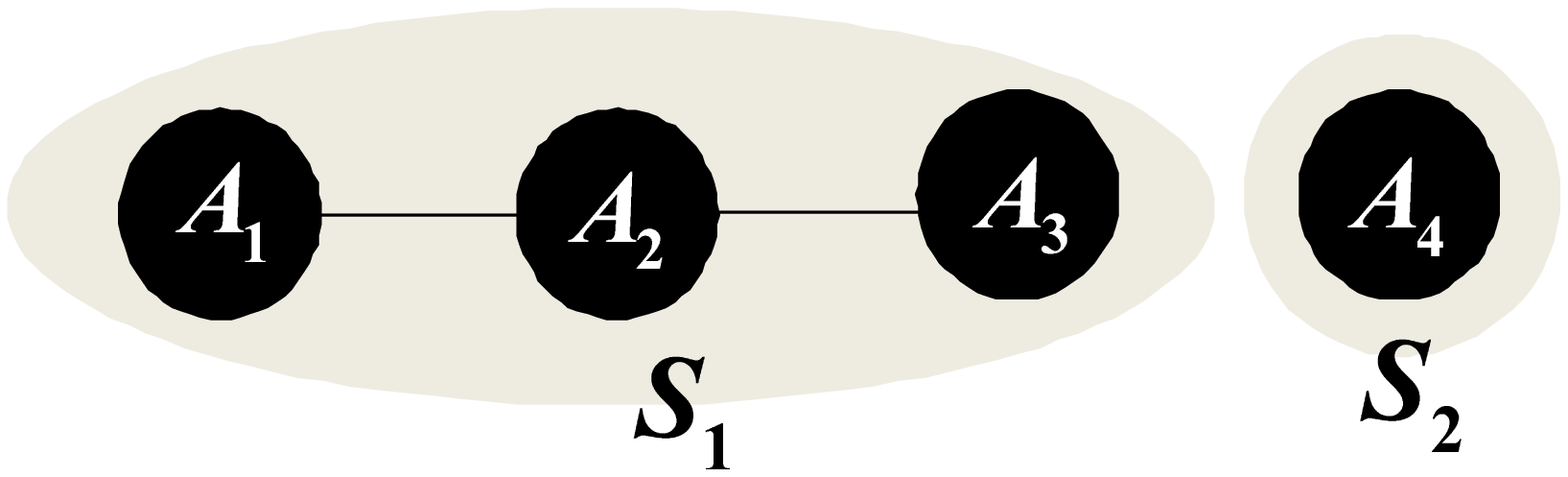}\\
  \caption{\label{fig}
Graph $G(\ket{\psi})$ when $\ket{\psi}$ is
$\ket{0000}+\ket{0110}+\ket{1200}+\ket{1310}+\ket{0001}+\ket{0111}+\ket{1201}+\ket{1311}$.}
\end{figure}

Using this graph, we obtain a partition of the parties $A_1,\ldots,
A_n$ and denote it by $\{S_1(\ket{\psi}),\ldots, S_k(\ket{\psi})\}$.
That is, $S_1(\ket{\psi}),\ldots, S_k(\ket{\psi})$ are subsets of
$\{A_1,\ldots, A_n\}$ and $\cup_{i=1}^k
S_i(\ket{\psi})=\{A_1,\ldots, A_n\}$. Besides, $S_i(\ket{\psi}) \cap
S_j(\ket{\psi})$ is empty when $i \neq j$. For any two parties $A_1$
and $A_2$ in $S_i(\ket{\psi})$, $A_1$ connects $A_2$ in the graph
$G(\ket{\psi})$, or there exist parties $A_{c_1}, \ldots, A_{c_l}$
in $S_i(\ket{\psi})$ such that $A_1$ connects $A_{c_1}$, $A_2$
connects $A_{c_l}$, and $A_{c_i}$ connects $A_{c_{i+1}}$ in the
graph $G(\ket{\psi})$. In the former case, we call $A_1$ connects
$A_2$ in the graph $G(\ket{\psi})$ directly while in the later case,
we call $A_1$ connects $A_2$ in the graph $G(\ket{\psi})$
indirectly. So, in other words, all elements of $S_i(\ket{\psi})$
are connected with each other at least indirectly. Further, any
element of $S_i(\ket{\psi})$ is not connected with any element of
$S_j(\ket{\psi})$ even indirectly when $i \neq j$.

In order to characterize the partition of the set $\{A_1,\ldots,
A_n\}$, we define the partial order among partitions. For two given
partitions of the set $\{A_1,\ldots, A_n\}$, $\{S_1,\ldots, S_k\}$
and $\{S_1',\ldots, S_l'\}$, we denote $\{S_1,\ldots, S_k\}\succ
\{S_1',\ldots, S_l'\}$ when any subset $S'_i$ in $\{S_1',\ldots,
S_l'\}$ belongs to a subset $S_i$ in $\{S_1,\ldots, S_k\}$. That is,
the partition $\{S_1',\ldots, S_l'\}$ is smaller than the partition
$\{S_1,\ldots, S_k\}$.

Since the pure state $\ket{\psi}$ has the form $\bigox_i \ket{\ps_i}_{S_i(\ket{\psi})}$,
we can classify multipartite states by using the partition of
the set $\{A_1,\ldots, A_n\}$.
\begin{align}
\cO(S_1,\ldots, S_k):=
\{
\ket{\psi}|
\{S_1(\ket{\psi}),\ldots, S_k(\ket{\psi})\}
=\{S_1,\ldots, S_k\}
\}.
\end{align}
From the above discussion,
we can show the "only if" part of the following theorem.

\begin{theorem}\label{hat1}
Any state in
$\cO(S_1',\ldots, S_l')$
can be generated by MCLOCC from
any state of
$\cO(S_1,\ldots, S_k)$
if and only if the relation $\{S_1,\ldots, S_k\}\succ \{S_1',\ldots, S_l'\}$ holds.
\end{theorem}

The "if" part can be shown by the following lemmas.
Using the above theorem,
for any partition $\{S_1,\ldots, S_l\}$ of the set $\{A_1,\ldots, A_n\}$,
the set $\cO(S_1,\ldots, S_l)$ is an MCLOCC equivalent class
and characterized as follows.
\begin{align}
\cO(S_1,\ldots, S_l)=
\cO(S_1)\otimes \ldots \otimes \cO(S_l).\label{haeq1}
\end{align}
In fact, for $S_1=\{A_{n_1}, \ldots,A_{n_l}  \}$, the set $\cO(S_1)$
is equal to the MCLOCC equivalent class generated by
$\ket{\ghz_{d:A_{n_1}, \ldots,A_{n_l}}} :=\sum_{i=1}^d|i\ldots
i\rangle_{A_{n_1}, \ldots,A_{n_l}}$, where $d \ge 2$. Hence, in the
multipartite system $A_1,\ldots, A_n$, the state $\ket{\ghz_{d:A_1,
\ldots, A_n}}$ is maximal in the sense of MCLOCC.

Now, we prepare lemmas for our proof of "if" part of Theorem
\ref{hat1}.

\begin{lemma}
When the tensor rank of
the state pure state $\ket{\psi}$ on $A_1, \ldots, A_n$
is $d$,
it can be generated by SLOCC from
the GHZ state $\ket{\ghz_{d:A_1, \ldots, A_n}}$.
\end{lemma}

\begin{proof}
There exist $n$ sets of non-normalized vectors
$\{\ket{a_{1,i}} \}_{i=1}^d, \cdots, \{\ket{a_{n,i}} \}_{i=1}^d$
such that $|\psi\rangle= \sum_{i=1}^d \ket{a_{1,i}} \ox \cdots \ox
\ket{a_{n,i}}$. So we can find the local operators
$A_i,i=1,\cdots,n$ such that $\bigox^n_{i=1}A_i\sum_{i=1}^d |i\cdots
i\rangle =|\psi\rangle$.
\end{proof}

Using this lemma,
for an arbitrary subset $S_i \subset \{A_1, \ldots, A_n \}$,
we can show that any state in $\cO(S_i)$ can be generated by MCLOCC
from the GHZ state of $S_i$.
So, it is sufficient to show that
the GHZ state of $S_i$ can be generated by MCLOCC from any state in $\cO(S_i)$.
In order to prove this argument, we prepare the following two lemmas.

\begin{lemma}\label{lh1a}
For given $m$ parties $A_1, \ldots, A_m$ and an integer $1\le n\le
m-1$, we denote the composite system $A_n$ and $A_{n+1}$ by $B$.
Then, the state $\ket{\ghz_{d:A_1, \ldots, A_{n-1},B,A_{n+2},\ldots,
A_{m} }}$ can be generated from $\ket{\ghz_{d:A_1, \ldots, A_{n}}}
\otimes \ket{\ghz_{d:A_{n+1},\ldots, A_{m}  }}$ by LOCC concerning
$A_1, \ldots, A_{n-1},B,A_{n+2},\ldots, A_{m}$.
\end{lemma}

\begin{proof}
Let $P_i$ be the projection to the subspace spanned by
$\ket{1}\ket{1+i}, \ldots, \ket{d}\ket{d+i}$, where $d+i \equiv i (
\mod d)$. First, we apply the measurement $\{P_i\}_{i=1}^d$ on the
system $B$, which is the composite system $A_n$ and $A_{n+1}$. When
we obtain the outcome $i$, we apply the local unitary $\ket{j} \to
\ket{j-i}$ on the systems $A_{n+2},\ldots, A_{m}$. Then, we obtain
the GHZ state on $A_1, \ldots, A_{n-1},B,A_{n+2},\ldots, A_{m}$.
\end{proof}

\begin{lemma}\label{lh2a}
When $A_1$ and $A_2$ are not independent for a pure multipartite
state $\ket{\psi}$, $\ket{\Psi_{2:A_1,A_2}}:=
\sum_{i=1}^2 |ii\rangle_{A_1,A_2}$ can be generated by
SLOCC from $\ket{\psi}$. Hence, when $A_1$ and $A_2$ are not
completely independent for a pure multipartite state $\ket{\psi}$,
$\ket{\Psi_{2:A_1,A_2}}$ can be generated by SLOCC from $\ket{\psi}$.
\end{lemma}

\begin{proof}
For any $i\ne 1,2$, there is a rank-one projection $P$ on the system
$A_i$ such that $A_1$ and $A_2$ are not independent for a pure
multipartite state $A_i \ket{\psi}$. Thus, the state $\ket{\psi}$
can be converted to a pure entangled state on $A_1$ and $A_2$ by
SLOCC. Using this state, we can generate $\ket{\ps_{2:A_1,A_2}}$ by
SLOCC.
\end{proof}

We choose an arbitrary pure state $\ket{\ps}$ in $\cO(S_i)$, where $S_i=\{A_{a_1}, \ldots, A_{a_k}\}$.
Using Lemma \ref{lh2a}, we can produce Bell states between two parties connected in the graph $G(\ket{\ps})$
by MCLOCC from $\ket{\ps}$.
Then, using Lemma \ref{lh1a}, we can produce the GHZ state
$\ket{\ghz_{2:A_{a_1}, \ldots, A_{a_k}}}$ from the above Bell states by LOCC.
Thus, we obtain the "if" part of Theorem \ref{hat1}.




\section{\label{sec:classificationSLOCC} classification of multipartite pure states under SLOCC}

In this section we study elementary step 2. Due to Eq.
(\ref{haeq1}), it is sufficient to consider the orbit decomposition
of the set $\cO(\{A_1, \ldots, A_n\})$ by SLOCC.

The problem has been widely studied in the past
decade~\cite{dvc00,ccm06,cds08,ccd10,cds10,bkm09} even in the
multipartite scenario. In particular, it has been pointed out that
deciding whether a GHZ state is SLOCC convertible into a
multipartite state is equal to the derivation of \emph{tensor rank}
\cite{cds08}.
This is the generalization of the Schmidt rank in
multipartite systems. Formally, for states in $N$-partite quantum
systems, each of which is described by a $d$-dimensional Hilbert
space $\h_i$ ($i=1,\ldots,N$), the tensor rank $\rk(\psi)$ of a
state $\ket{\psi} \in \bigotimes_{\alpha=1}^N \h_\alpha$, defined as
the smallest number of product states
$\{\bigotimes_{\alpha=1}^N\ket{\phi^\alpha_i}\}_{i=1...rk(\psi)}$
whose linear span contains $\ket{\psi}$.  The tensor rank has been
extensively studied in algebraic complexity theory
\cite{Kruskal-1977a, buergisser-book}. While it is easy to compute
for $N=2$ (Schmidt rank), even for $N=3$, determining the rank of a
state is NP-hard \cite{Haastad-1990a}.  This is one reason why SLOCC
convertibility in multipartite systems is so challenging.

It's well known that, classifying multipartite states in high
dimensions is difficult for the number of parameters increases
exponentionally with local ranks. Different from the traditional
way, we will use tensor rank to make the problem more tractable. The
classification with both tensor rank and local ranks provides a new
perspective on the structure of multipartite states under SLOCC. In
particular we will rely on the following result.
 \bl
 \label{le:fourinvariantsSLOCC}
 Two SLOCC-equivalent pure states on the multipartite system $A_1,\ldots, A_N$
have to have identical SLOCC invariants, $\rk, r_{A_1}, \ldots, r_{A_N}$.
So inequivalent SLOCC orbits have to have at least one
 non-identical invariant out of $\rk, r_{A_1}, \ldots, r_{A_N}$.
Besides, the $N+1$ SLOCC invariants are non-increasing under non-invertible SLOCC.
 \el

The lemma is obvious. Based on it, we classify multipartite pure
states into a few inequivalent SLOCC orbits and build a few theorems
to characterize them.
First of all we divide the orbit
$\cO(\{A_1, \ldots, A_n\})$
into a few \emph{principal} SLOCC-closed sets
in the light of only tensor rank
 \bea
 \label{ea:MCLOCCtoprincipal}
\cO(\{A_1, \ldots, A_n\}) = \bigcup_{d \ge 2} \cR(d),
 \eea
where $\cR(d)$ is
the set of states with tensor rank $d$ in $\cO(\{A_1, \ldots, A_n\})$,
which is closed under reversible SLOCC.
By using
Lemma~\ref{le:fourinvariantsSLOCC}, it suffices to consider only one
principal SLOCC orbit, namely $\cR(d)$ with constant
$d\ge2$. Taking into account other $N$ SLOCC invariants $r_{A_1}, \ldots, r_{A_N}$,
we can split this SLOCC-closed set into \emph{sub} SLOCC-closed sets
denoted as
 \bea
  \label{ea:princialtosub}
 \cR(d)
 &=&
\bigcup_{d\ge r_{A_1}, \ldots, r_{A_N} }  \cR(d\| r_{A_1}, \ldots, r_{A_N}) ,
 \eea
where $\cR(d\|r_{A_1}, \ldots, r_{A_N})$ is the set of states with
constant tensor rank $d$ and local ranks $r_{A_1}, \ldots, r_{A_N}$.
Evidently
there are many states fulfilling the property. In this sense we can
decompose the sub SLOCC-closed set into SLOCC orbits of
SLOCC-inequivalent states
 \bea
   \label{ea:subtoorbit}
 \cR(d\| r_{A_1}, \ldots, r_{A_N} )
  =
 \bigcup_j
 \cO_{SLOCC} \bigg(\sum^d_{i=1} \ket{a_{1,i}^j,\ldots, a_{N,i}^j} \bigg),
 \nonumber\\
 \eea
where the ranks of the spaces spanned by
$\{\ket{a_{1,i}^j}\},\ldots, \{\ket{a_{N,i}^j}\}$, and
$\{\ket{a_{1,i}^j,\ldots, a_{N,i}^j}\}$ are $r_{A_1},\ldots,r_{A_N}$, and $d$,
respectively.
In particular, when $r_{A_1}$ is $d$, the states
$\{\ket{a_{1,i}^j}\}$ can be restricted to the states $\{\ket{i}\}$.
This fact can be applied to the cases of $r_{A_j}=d$.

By using the three steps in
Eqs.~\ref{ea:MCLOCCtoprincipal},~\ref{ea:princialtosub} and
\ref{ea:subtoorbit}, we have given a clear hierarchy of fully
entangled multipartite states. In each level states from different
sets/orbits are SLOCC-inequivalent. From this hierarchy we see that
the third level Eq.~\ref{ea:subtoorbit}, namely splitting the sub
SLOCC-closed set into SLOCC orbits is a key step to classify
multipartite states in each principal SLOCC-closed set. This is the
topic of subsection~\ref{subsec:decomposition of
subSLOCCclosedsets}.

\subsection{\label{subsec:decomposition of subSLOCCclosedsets}
decomposition of sub SLOCC-closed sets}

We start by characterizing the simplest sub SLOCC-closed set in
$\cR(d)$. Note that we don't normalize the states
because it doesn't influence the equivalence and partial order under
SLOCC and LOCC.


 \bl
 \label{le:subSLOCCorbit,1}
For any pure state $|\psi\rangle$ with
$\rk(\psi)=r_{A_1}(\psi)=\cdots =r_{A_N}(\psi)=d$, there exists invertible
matrices $X_1,\ldots,X_N$ such that $X_1\otimes \cdots \otimes X_N
\sum_{i=1}^d |i\cdots i\rangle_{A_1\ldots A_N} =|\psi\rangle$.
So we have
 \bea
 \label{ea:subSLOCCorbit,1}
  \cR(d\|d,d,d)
  &=&
  \cO_{SLOCC} \bigg(\ket{{\ghz}_{d:A_1\ldots A_N}} \bigg).
 \eea
The $d$-dimensional GHZ state $\ket{{\ghz}_{d:A_1\ldots A_N}}$ is the
\emph{canonical state} of the orbit.
 \el
 \bpf
There exist $N$ sets of non-normalized linearly independent vectors
$\{\ket{a_{1,i}} \}_{i=1}^d, \cdots, \{\ket{a_{N,i}} \}_{i=1}^d$
such that $|\psi\rangle= \sum_{i=1}^d \ket{a_{1,i}} \ox \cdots \ox
\ket{a_{N,i}}$. So we can find the invertible local operators
$A_i,i=1,\cdots,N$ such that $\bigox^N_{i=1}A_i\sum_{i=1}^d |i\cdots
i\rangle =|\psi\rangle$.
 \epf

However, other sub SLOCC-closed subsets are not SLOCC orbits of
given states. To show this fact, we recall the classification of
$2\times3\times3$ states under SLOCC \cite{ccm06}.
 \bl
 There are six inequivalent states in the
 $2\times3\times3$ space under SLOCC
 \bea
 \left|\Psi_1\right\rangle &=& \left|000\right\rangle+\left|111\right\rangle+(\left|0\right\rangle+\left|1\right\rangle)\left|22\right\rangle,
 ~\rk=3,\\
 \left|\Psi_2\right\rangle &=& \left|010\right\rangle+\left|001\right\rangle+\left|112\right\rangle+\left|121\right\rangle,
 ~\rk=4,\\
 \left|\Psi_3\right\rangle &=& \left|000\right\rangle+\left|111\right\rangle+\left|022\right\rangle,
 ~\rk=3,\\
 \left|\Psi_4\right\rangle &=& \left|100\right\rangle+\left|010\right\rangle+\left|001\right\rangle+\left|112\right\rangle+\left|121\right\rangle,
 ~\rk=4,\nonumber\\
 \\
 \left|\Psi_5\right\rangle &=& \left|100\right\rangle+\left|010\right\rangle+\left|001\right\rangle+\left|022\right\rangle,
 ~\rk=4,\\
 \left|\Psi_6\right\rangle &=&
 \left|100\right\rangle+\left|010\right\rangle+\left|001\right\rangle+\left|122\right\rangle,
 ~\rk=4.
 \eea
 \el
So the $2\times3\times3$ state belongs to two principal SLOCC-closed
sets $ \cR(4) $ and $ \cR(3) $. More precisely, they
belong to only two sub SLOCC orbits respectively, i.e.,
 \bea
 \cR( 3\|2,3,3 )
 &=&
 \bigcup_{i=1,3}
 \cO_{SLOCC} \bigg(\ket{\Ps_i} \bigg),\label{Haya1}
\\
 \cR( 4\|2,3,3 )
 &=&
 \bigcup_{i=2,4,5,6}
 \cO_{SLOCC} \bigg(\ket{\Ps_i} \bigg).
 \eea

By removing the parameters with locally invertible operators, we can
generalize Eq. \ref{Haya1} to the $N$-partite states such as
 \bea
 && \cR( d\|d-1,d, \cdots, d )
 \nonumber\\
 =&&
 \bigcup_{j=1}^{d-1}
 \cO_{SLOCC}
 \bigg(
 \sum_{i=1}^{d-1}\ket{i}^{\ox N}
 +
 \sum_{i=1}^j \ket{i} \ket{d}^{\ox (N-1)}
 \bigg).
 \eea
Note that the $d-1$ states in the big parentheses are
SLOCC-inequivalent. The reason is that they can be written as
 \bea
 \ket{\ph_j}
 &=&
 \sum_{i=1}^{j} \ket{i} (\ket{i}^{\ox (N-1)} + \ket{d}^{\ox (N-1)})
 +
 \sum_{i=j+1}^{d-1} \ket{i}^{\ox N}
 \nonumber\\
 &=&
 \sum_{i=1}^{d-1} \ket{a_{i,j}} \ket{\ps_{i,j}},
 \eea
where $\ket{a_{i,j}} \in \cH_1$ and $\ket{\ps_{i,j}} \in
\bigox^{N}_{i=2} \cH_i$. Because the state in the range of $\r_{2,
\cdots, N}$ has the expression $\sum_{i=1}^{j} a_i (\ket{i}^{\ox
(N-1)} + \ket{d}^{\ox (N-1)}) + \sum_{i=j+1}^{d-1} a_i \ket{i}^{\ox
(N-1)}$, which is entangled if any coefficient $a_i \ne 0, i \in
[1,j]$. Hence there are at least $j$ entangled states
$\ket{\ps_{i,j}}$ in the state $\ket{\ph_j}$. Because entanglement
cannot be changed under locally invertible operators, the states
$\ket{\ph_j}$ cannot inter-convertible under SLOCC.

In general the decomposition of sub SLOCC-closed sets is not easy.
As a partial result, we can use the method in~\cite{ccm06,cmy10} to
classify $2\times M \times N$ states under SLOCC. Besides, tensor
rank is also computable effectively for this family of
states~\cite{jaja78,cmy10}. These results help carry out the
hierarchy in Eq. \ref{ea:MCLOCCtoprincipal}-\ref{ea:subtoorbit}.
First by following Eq.~\ref{ea:MCLOCCtoprincipal}, we distribute the
$2\times M \times N$ states into the principal SLOCC-closed sets
$\cR(d)$. Next by following Eq.~\ref{ea:princialtosub}, we
can decompose each $\cR(d)$ into sub SLOCC-closed sets to
which $2\times M \times N$ states belong. Third we can decide the
SLOCC orbits in each sub SLOCC-closed sets, which realizes
Eq.~\ref{ea:subtoorbit}.

One can similarly characterize $L \times M \times N (L \ge 3)$ pure
states. However the problem becomes complex as the dimension
increases, e.g., there are infinitely many inequivalent
$2\times4\times4$ and $3\times3\times3$ states under SLOCC
\cite{ccm06}. Actually it's been shown that the classification of
tripartite pure states is a NP-hard problem~\cite{cds08}. In this
context, it becomes quite important to figure out a general and
clear configuration of the orbit $\cO_{MCLOCC}^{ABC}$. This can be
done through the SLOCC partial order between the principal and sub
SLOCC-closed sets, which is the topic of next subsection.

\subsection{\label{subsec:partialorder,SLOCCorbits}
SLOCC partial order}

We have introduced three kinds of SLOCC classification for
multipartite states in previous paragraphs, namely the principal
SLOCC-closed set, the sub SLOCC-closed set, and the SLOCC orbit. To
get a clear configuration of the orbit $\cO(A_1, \ldots,A_N)$, we will
investigate the orbits successively embedded into each other, i.e.,
$  \cO_{SLOCC}(\ket{\ps} ) \subset \cR(d\|r_{A_1},\ldots, r_{A_N} )
\subset \cR(d) \subset \cO(A_1, \ldots,A_N). $ In other words,
the former orbit or set always belong to the same latter orbit or
set under SLOCC partial order, respectively. It will also help
reduce the question of rapidly increasing parameters in high
dimensions.

First, we study the partial order of different principal
SLOCC-closed set, $\cR(d_1)$ and $\cR(d_2),
~d_1>d_2$. There are a few states in the former which can be
converted into the latter via SLOCC. For example, we can readily
realize $\ket{{\ghz}_{d_1:A_1,\ldots, A_N}} \succ_{SLOCC} \ket{\ps}$ where $\forall
\ket{\ps} \in \cR(d_2)$. However for generic states, we
cannot carry out the transformation as above. The biggest reason is
that there are states in the SLOCC-closed set  $\cR(d_1)$,
which have smaller local rank than that of some states in the
SLOCC-closed set $\cR(d_2)$, and local ranks cannot be
increased under SLOCC as mentioned in
Lemma~\ref{le:fourinvariantsSLOCC}. For example, the states
$\ket{111}+\ket{122}+\ket{213}+\ket{224}$ and $\ket{{\ghz}_3}$ have
tensor rank four and three, respectively. However they are evidently
SLOCC incomparable. On the other hand, the transformation
$\cR(d_2) \succ_{SLOCC} \cR(d_1)$ is forbidden due
to the fact that tensor rank is non-increasing under SLOCC. So we
conclude
 \bt
 Generic states in different
 principal SLOCC-closed sets are SLOCC incomparable.
 That is, for $\ket{\ps_1} \in \cR(d_1)$ and $\ket{\ps_2} \in \cR(d_2),
 ~d_1>d_2$ we have $\ket{\ps_1} \ncong_{SLOCC} \ket{\ps_2}$. So we
 readily get the corollary that for generic states
 $\ket{\ps_1} \ncong_{LOCC} \ket{\ps_2}$.
 \et

Second, we study the partial order of different sub SLOCC-closed sets
from the principal SLOCC-closed set  $\cR(d)$.
We present the following result.
 \bt
 \label{thm:ghztoprincipal,partialorder=slocc}
The GHZ state with rank $d$ is maximal concerning SLOCC partial
order among pure states with tensor rank $d$. That is, the GHZ state
can be used to generate all states in the principal SLOCC-closed
sets $\cR(d)$ with some nonzero probability,
 \bea
 \ket{{\ghz}_{d:A_1,\ldots, A_N}}
 \succ_{SLOCC}
 \ket{\psi},~~ \forall \ket{\psi} \in \cR(d).
 \eea
 \et
 \bpf
 Similar to
 Lemma~\ref{le:subSLOCCorbit,1}.
 \epf

Third, we can easily indicate the partial order between SLOCC orbits
$\cO_{SLOCC}(\ket{\ps_i})$ that form a sub SLOCC-closed set
$\cR(d\| r_{A_1}, \ldots, r_{A_N})$. In the light of
Lemma~\ref{le:fourinvariantsSLOCC}, two SLOCC-inequivalent states
with identical four SLOCC invariants are SLOCC incomparable, and
thus also LOCC incomparable. In last subsection, we have introduced
the way to characterize SLOCC-inequivalent states in $2\times
M\times N$ space. More efforts are still required for high
dimensions.

As a short summary, in this section we have classified multipartite
states based on Lemma~\ref{le:fourinvariantsSLOCC}. In our hierarchy
of three levels, we decomposed the MCLOCC orbit $\cO_{MCLOCC}^{ABC}$
into principal SLOCC-closed sets, from each of which we have
proposed sub SLOCC-closed sets and decomposed them into inequivalent
SLOCC orbits. We also have given the partial order under SLOCC. Our
method is generally different from, and also greatly extends the
previous methods of classifying multipartite states
\cite{ccm06,ccd10}.

In next section, we will extend our results to the LOCC
classification and show the LOCC partial order between the orbits.

\section{\label{sec:classificationLOCC} LOCC partial order}

In this section we carry out elementary step 3 for LOCC partial
order. In the bipartite case, the state $\sum_{i=1}^d\ket{i,i}$ is
called the maxially entangled state because the state
$\sum_{i=1}^d\ket{i,i}$ is maximal concerning the LOCC partial order
among pure bipartite states with rank $d$. Intuitively, we might
expect that the state $\ket{\ghz_d}$ can generate any pure state
with tensor rank $d$ by LOCC in the multipartite system. However, it
is not true even in the tripartite case with $\rk=r_A=r_B=r_C=2$,
because the state
$\ket{0,0,0}+(\ket{0}+\ket{1})(\ket{0}+\ket{1})(\ket{0}+\ket{1})$
cannot be generated from the state $\ket{\ghz_2}$ by LOCC
\cite[Theorem 5]{tgp10}. This fact implies that there is no maximal
element concerning LOCC even among the SLOCC orbit $\cO_{SLOCC}
\bigg(\sum_{i=1}^d |iii\rangle \bigg)$.

In order to treat the multipartite extension of maximally entangled
states, we introduce a new class of multipartite pure states. A
multipartite pure state $\ket{\psi}$ is \emph{reduced separable}
when there exists a party, concerning which, the partial trace of
$\ket{\psi}$ is fully separable, namely the convex sum of fully
factorized states $\bigox_i \proj{a_i}$. In other words a reduced
separable pure state is the purification of a fully separable state.
To characterize the latter, we can use the \emph{cardinality}
$c(\r)$ which is defined as the minimal number of product pure
components required to construct the fully separable state $\r$
\cite{lockhart00,as10}
\begin{align}
  \label{al:cardinality}
  \r = \sum^{c(\r)}_{i=1} p_i \proj{a^{(1)}_i,a^{(2)}_i,\ldots,a^{(N-1)}_i}.
\end{align}
Let its purification be the state $\ket{\r}$ such that $\tr_N
\proj{\r}=\r$. 
For convenience, we also call the above number the cardinality of
the purification, i.e., $c(\r) = c(\ket{\r})$. Note that the tensor
rank of the GHZ state coincides with its cardinality, we obtain the
following theorem.


 \bt
\label{thm:ghztoghx,locc} In the $N$-partite system, the GHZ state
with rank $d$ can generate any reduced separable pure state with the
cardinality not bigger than $d$ by LOCC. That is, the GHZ state with
rank $d$ is maximal concerning LOCC among reduced separable pure
states with the largest cardinality $d$, i.e.,
 \bea
 \label{thm:ghztoghx,locc}
 \ket{{\ghz}_d}
 \succ_{LOCC}
 \sum^d_{i=1} \sqrt{p_i} \ket{a^{(1)}_i,a^{(2)}_i,\ldots,a^{(N-1)}_i ,i}.
 \eea
 \et
 \bpf
 First by using the complete POVM $\{\sum^d_{j=1} \sqrt{p_j} \ketbra{j}{j \oplus k}\}, k=1,\cdots, d$
 and local unitary operations, we can realize $\ket{{\ghz}_d}
 \succ_{LOCC} \sum^d_{i=1} \sqrt{p_i} \ket{i}^{\ox N}$. Second
 we consider the complete POVM on system $A_1$ such that
 $$P_k=\sum^d_{j=1} e^{\frac{2\pi I}{d}j k} \ketbra{a^{(1)}_j}{j} $$,
 which satisfy $\sum^d_{k=1} P_k^\dg P_k = I$. Then we get always one of the
 following states
 $$\sum^d_{j=1} e^{\frac{2\pi I}{d}j k} \sqrt{p_j} \ket{a^{(1)}_j} \ket{j}^{\ox (N-2)}
 \ket{j},~k=1, \cdots, N.$$ Starting from this state, we perform the
 complete POVM $\{ \sum^d_{j=1} e^{\frac{2\pi I}{d}j k}
 \ketbra{a^{(2)}_j}{j} \}$ on system $A_2$, to get the state
 $\sum^d_{j=1} e^{I \t_j} \sqrt{p_j} \ket{a^{(1)}_j,a^{(2)}_j} \ket{j}^{\ox (N-3)}
 \ket{j}.$ By using the iterative method we can finally get the
 state
 $$ \sum^d_{i=1} e^{I \t_j'} \sqrt{p_i} \ket{a^{(1)}_i,a^{(2)}_i,\ldots,a^{(N-1)}_i ,i},$$
 which can be recovered to the asserted state via a phase gate on
 system $A_N$. This completes the proof.
 \epf
One may similarly get the LOCC partial order for more tripartite
states, e.g., to transform $\ket{{\ghz}_d}$ into
$\sum^d_{i=1}\ket{a_i,b_i,c_i}$. Note that the partial order is
usually non-invertible, because two states are LOCC-interconvertible
if and only if they are unitarily equivalent \cite{bpr00}.


\section{\label{sec:conclusion} conclusions}

In this paper, we have built the concepts of MCLOCC and MCSLOCC and
showed their relations to the LOCC and SLOCC environment. These
environments form the criteria in our classification hierarchy for
multipartite states. In particular, We have classified the
SLOCC-equivlant states by using tensor rank which is a basic concept
from algebraic complexity. Besides, we have derived the SLOCC and
LOCC partial order of different sets and orbits from the hierarchy.
Our method is essentially different from the previous ones in
literature, which relies only on the local ranks. The results are
important for understanding the convertibility between multipartite
states. A further direction from this paper is to characterize the
conversion between states with identical tensor rank under LOCC.
This may give a better interpretation of tensor rank for studying
multipartite states.



\end{document}